\documentclass[a4paper,twocolumn,11pt,accepted=2023-09-15]{quantumarticle}
\pdfoutput=1 
\usepackage{dsfont} 
\usepackage{comment} 
\usepackage[numbers,sort&compress]{natbib} 	
\usepackage[utf8]{inputenc} 
\usepackage[T1]{fontenc} 
\usepackage[english]{babel} 
\usepackage[dvipsnames,x11names]{xcolor} 
\usepackage{amsmath,amssymb,amsthm,bm,amsfonts,bbm} 
\usepackage{mathtools} 
\usepackage[colorlinks=true,citecolor=Green,linkcolor=Purple,urlcolor=Blue]{hyperref} 

\usepackage{physics}
\usepackage{extarrows}
\usepackage{enumitem}

\DeclareMathOperator{\sgn}{sgn}
\newtheorem{lemma}{Lemma}

\newtheorem{protocol}{Protocol}

\begin{document}
\title{The minimal communication cost for simulating entangled qubits
}

\author{Martin J. Renner,}
    \email{martin.renner@univie.ac.at}
    \orcid{0000-0002-3408-6848}
    \affiliation{University of Vienna, Faculty of Physics, Vienna Center for Quantum Science and Technology (VCQ), Boltzmanngasse 5, 1090 Vienna, Austria}
    \affiliation{Institute for Quantum Optics and Quantum Information (IQOQI), Austrian Academy of Sciences, Boltzmanngasse 3, 1090 Vienna, Austria}

\author{Marco Túlio Quintino}
    \email{Marco.Quintino@lip6.fr}
    \orcid{0000-0003-1332-3477}
    \affiliation{Sorbonne Universit\' {e}, CNRS, LIP6, F-75005 Paris, France}
    \affiliation{Institute for Quantum Optics and Quantum Information (IQOQI), Austrian Academy of Sciences, Boltzmanngasse 3, 1090 Vienna, Austria}
    \affiliation{University of Vienna, Faculty of Physics, Vienna Center for Quantum Science and Technology (VCQ), Boltzmanngasse 5, 1090 Vienna, Austria}


\begin{abstract}
We analyze the amount of classical communication required to reproduce the statistics of local projective measurements on a general pair of entangled qubits,  ${\ket{\Psi_{AB}}=\sqrt{p}\ket{00}+\sqrt{1-p}\ket{11}}$ ($1/2\leq p \leq 1$). We construct a classical protocol that perfectly simulates local projective measurements on all entangled qubit pairs by communicating one classical trit. Additionally, when $\frac{2p(1-p)}{2p-1} \log{\left(\frac{p}{1-p}\right)}+2(1-p)\leq1$,
approximately $0.835 \leq p \leq 1$, we present a classical protocol that requires only a single bit of communication. The latter model even allows a perfect classical simulation with an average communication cost that approaches zero in the limit where the degree of entanglement approaches zero ($p \to 1$). This proves that the communication cost for simulating weakly entangled qubit pairs is strictly smaller than for the maximally entangled one. 
\end{abstract}

\maketitle

\section{Introduction}
Bell's nonlocality theorem \cite{Bell} shows that quantum correlations cannot be reproduced by local hidden variables. This discovery has significantly changed our understanding of quantum theory and correlations allowed by nature. Additionally, Bell nonlocality found application in cryptography~\cite{Ekert91} and opened the possibility for protocols in which security can be certified in a device-independent way~\cite{Acin07,pironio10,Vazirani14,supic19}.

Since quantum correlation cannot be explained by local hidden variables it is interesting to ask which additional resources are required to reproduce them. For instance, can the statistics of local measurements on two entangled qubits be simulated if the local hidden variables are augmented with some classical communication?
%
However, since measurements are described by continuous parameters, one might expect that the communication cost to reproduce these quantum correlations is infinite \cite{Maudlin1992}. After a sequence of improved protocols for entangled qubits \cite{Brassard1999, Steiner2000, cerf2000, Pati2000, massar2001}, a breakthrough was made by Toner and Bacon in 2003 \cite{tonerbacon2003}.
They showed that a single classical bit of communication is sufficient to simulate the statistics of all local projective measurements on a maximally entangled qubit pair. Classical communication has then been established as a natural measure of Bell nonlocality~\cite{Bacon2003, degorre2005, Degorre2007,  Regev2010, Branciard2011, Branciard2012, Maxwell2014, Brassard2015, Brassard2019, Zambrini2019} and found applications in constructing local hidden variable models~\cite{degorre2005}.

For non-maximally entangled qubit pairs, somehow counterintuitively, all known protocols require strictly more resources. In terms of communication, the best-known result is also due to Toner and Bacon~\cite{tonerbacon2003}. They present a protocol for non-maximally entangled qubits, which requires two bits of communication (see Ref.~\cite{Renner2022} for a two-bit protocol that considers general POVM measurements). The asymmetry of partially entangled states and other evidence suggested that simulating weakly entangled states may be harder than simulating maximally entangled ones. For instance, in Ref.~\cite{brunner2005} the authors prove that at least two uses of a PR-box are required for simulating weakly entangled qubit pairs. At the same time, a single use of a PR-box is sufficient for maximally entangled qubits~\cite{cerfprbox}. Additionally, weakly entangled states are strictly more robust than maximally entangled ones when the detection loophole is considered~\cite{Eberhard92,Cabello2007detection,Brunner2007detection,araujo12}.

In this work, we present a protocol that simulates the statistics of arbitrary local projective measurements on weakly entangled qubit pairs with only a single bit of communication. Then, we construct another protocol to simulate local projective measurements on any entangled qubit pair at the cost of a classical trit. 
\begin{figure}
    \centering
    \includegraphics[width=0.45 \textwidth]{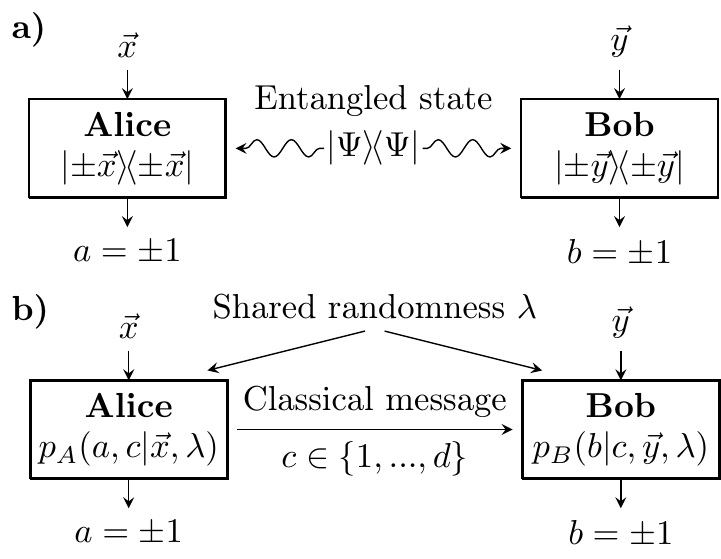}
    \caption{\textbf{a)} Alice and Bob perform local projective measurements on an entangled qubit pair. \textbf{b)} Classical scenario where Alice can send a classical message to Bob.
    }
    \label{fig1}
\end{figure}


\section{The task and introduction of our notation} 
Up to local unitaries, a general entangled qubit pair can be written as
\begin{align}
    \ket{\Psi_{AB}}=\sqrt{p}\ket{00}+\sqrt{1-p}\ket{11} \, , \label{statedef}
\end{align}
where $1/2 \leq p \leq 1$.
At the same time, projective qubit measurements can be identified with a normalized three-dimensional real vector $\vec{x}\in\mathbb{R}^3$, the Bloch vector, $\vec{x}=(x_x,x_y,x_z)$ (with $|\vec{x}|= 1$) via the equation $\ketbra{\vec{x}}=\big(\mathds{1} + \vec{x}\cdot\vec{\sigma}\big)/2$. Here, $\vec{\sigma}=(\sigma_X,\sigma_Y,\sigma_Z)$ are the standard Pauli matrices. In this way, we denote Alice's and Bob's measurement projectors as 
$\ketbra*{\pm\vec{x}}$ and $\ketbra*{\pm \vec{y}}$, which satisfy $\ketbra*{+\vec{x}}+\ketbra*{-\vec{x}}=\ketbra*{+\vec{y}}+\ketbra*{-\vec{y}}=\mathds{1}$.
According to Born's rule, when Alice and Bob apply their measurements on the entangled state $\ket{\Psi_{AB}}$, they output $a,b\in\{-1,+1\}$ according to the statistics:
\begin{align} \small
p_Q(a,b|\vec{x},\vec{y})=\Tr[\ketbra*{a\vec{x}}\otimes \ketbra*{b \vec{y}}\ \ketbra*{\Psi_{AB}}] . \label{prob}
\end{align} \normalsize

In this work, we consider the task of simulating the statistics of Eq.~\eqref{prob} with purely classical resources. More precisely, instead of Alice and Bob performing measurements on the actual quantum state, Alice prepares an output $a$ and a message $c \in \{1,2,...,d\}$ that may depend on her measurement setting $\vec{x}$ and a shared classical variable $\lambda$ that follows a certain probability function $\rho(\lambda)$ (see Fig.~\ref{fig1}~b)). Therefore, we can denote Alice's strategy as $p_A(a,c|\vec{x},\lambda)$. Afterwards, Alice sends the message $c$ to Bob, who produces an outcome $b$ depending on the message $c$ he received from Alice, his measurement setting $\vec{y}$, and the shared variable $\lambda$. In total, we denote his strategy as $p_B(b|c,\vec{y},\lambda)$. We want to remark that in our setting, Alice has no knowledge about Bob's measurement and vice versa. Therefore her strategy cannot depend on $\vec{y}$ and his strategy cannot depend on $\vec{x}$. Altogether, the total probability that Alice and Bob output $a,b\in\{-1,+1\}$ becomes:
\begin{align}
\begin{split}
p_C&(a,b|\vec{x},\vec{y})\\
&=	 \int_\lambda   \text{d}\lambda \; \rho(\lambda) \sum_{c=1}^{d}  p_A(a,c|\vec{x},\lambda) p_B (b|\vec{y},c,\lambda) \, . 
\end{split}
\end{align}
The simulation is successful if, for any choice of projective measurements and any outcome, the classical statistics match the quantum predictions:
\begin{equation}
p_C(a,b|\vec{x},\vec{y})=p_Q(a,b|\vec{x},\vec{y}) \, .
\end{equation}
We want to remark that the roles of Alice and Bob are interchangeable due to the symmetry of the state $\ket{\Psi_{AB}}$ in Eq.~\eqref{statedef}. Therefore, any protocol of this work can be rewritten into a protocol where Bob communicates a message (of the same length) to Alice.

For what follows, we also introduce the Heaviside function, defined by $H(z)=1$ if $z\geq 0$ and $H(z)=0$ if $z<0$, as well as the related functions $\Theta(z):=H(z)\cdot z$ and the sign function $\sgn(z):=H(z)-H(-z)$.

\section{Revisiting known protocols}
Our methods are inspired by the best previously known protocol to simulate general entangled qubit pairs, the so-called "classical teleportation" protocol \cite{cerf2000, tonerbacon2003}. To understand the idea, we first rewrite the quantum probabilities in Eq.~\eqref{prob} by using the rule of conditional probabilities $p(a,b|\vec{x},\vec{y})=p(a|\vec{x},\vec{y})\cdot p(b|\vec{x},\vec{y},a)$. More precisely, we denote with $p_\pm:=
\sum_{b} p(a=\pm 1, b|\vec{x},\vec{y})$ the marginal probabilities of Alice's output that read as follows:
\begin{align}
    p_\pm =\Tr[\ketbra*{\pm\vec{x}}{\pm\vec{x}}\otimes \mathds{1}\ \ketbra*{\Psi_{AB}}{\Psi_{AB}}]\, . \label{pplusminus}
\end{align}
Note that, due to non-signalling, the marginals $p_\pm$ do not depend on $\vec{y}$. At the same time, given Alice's outcome $a=\pm 1$, Bob's qubit collapses into a pure post-measurement state, that we denote here as:
\begin{align}
    \ketbra*{\vec{v}_\pm}&:=\Tr_A[\ketbra*{\pm \vec{x}}\otimes \mathds{1}\ \ketbra*{\Psi_{AB}}]/p_\pm \, . \label{vplusminus}
\end{align}
If now Bob measures his qubit with the projectors $\ketbra{\pm \vec{y}}$, he outputs $b$ according to Born's rule: \begin{align}
    p(b|\vec{x},\vec{y},a)=\Tr [\ketbra{b \vec{y}} \ketbra{\vec{v}_a}]= |\braket{\vec{v}_a}{b\vec{y}}|^2 \, .
\end{align}
With the introduced notation, we can rewrite the quantum probabilities from Eq.~\eqref{prob} into: 
\begin{align}
    p_Q(a,b|\vec{x},\vec{y})=p_{a}\cdot |\braket{\vec{v}_a}{b\vec{y}}|^2 \, . \label{rewrite}
\end{align}
This directly implies a strategy to simulate entangled qubit pairs. Alice outputs $a=\pm 1$ according to her marginals $p_\pm$. Then, given her outcome $a$, she prepares a qubit in the correct post-measurement state $\ketbra{\vec{v}_a}$ and sends it to Bob. Finally, he measures the qubit with his projectors $\ketbra{\pm \vec{y}}$.

However, in a classical simulation, Alice cannot send a physical qubit to Bob. Nevertheless, it is possible to simulate a qubit in that prepare-and-measure (PM) scenario with only two classical bits of communication \cite{tonerbacon2003}. In order to do so, Alice and Bob share four normalized three-dimensional vectors $\vec{\lambda}_1,\vec{\lambda}_2, \vec{\lambda}_3, \vec{\lambda}_4\in S_2$. The first two $\vec{\lambda}_1$ and $\vec{\lambda}_2$ are uniformly and independently distributed on the sphere, whereas $\vec{\lambda}_3=-\vec{\lambda}_1$ and $\vec{\lambda}_4=-\vec{\lambda}_2$. From these four vectors, Alice chooses the one that maximizes $\vec{\lambda}_i\cdot \vec{v}_a$ and communicates the result to Bob. This requires a message with four different symbols ($d=4$), hence, two bits of communication. It turns out that the distribution of the chosen vector becomes $\Theta(\vec{v}_a \cdot \vec{\lambda})/\pi$ (see Appendix~\ref{appendixb} for an independent proof). Finally, Bob takes the chosen vector $\vec{\lambda}$ and outputs $b=\sgn(\vec{y}\cdot \vec{\lambda})$. This precisely reproduces quantum correlations as specified by the following Lemma (see Appendix~\ref{prooflemma} for a proof):
\begin{lemma}\label{lemma1}
Bob receives a vector $\vec{\lambda}\in S_2$ distributed as $\rho(\vec{\lambda})=\Theta(\vec{v}\cdot \vec{\lambda})/\pi$ and outputs $b=\sgn(\vec{y}\cdot \vec{\lambda})$. For every qubit state $\vec{v}\in S_2$ and measurement $\vec{y}\in S_2$ this reproduces quantum correlations:
\begin{align}
    p(b=\pm 1|\vec{y},\vec{v})&=(1 \pm \vec{y}\cdot \vec{v})/2=|\braket{\pm \vec{y}}{\vec{v}}|^2\, .
\end{align}
\end{lemma}
That the distribution $\Theta(\vec{\lambda} \cdot \vec{v}_a)/\pi$ serves as a classical description of the qubit state $\ketbra{\vec{v}_a}$ was already observed by Kochen and Specker \cite{KochenSpecker1967}. Later, this Kochen-Specker model was used for the task of simulating qubit correlations, see e.g. Ref.~\cite{gisingisin1999, cerf2000, degorre2005}.

\section{Our approach} 
The previous approach of using the prepare-and-measure scenario to simulate entangled qubits~\cite{cerf2000, tonerbacon2003} has a natural limitation. In fact, simulating a qubit in a PM scenario requires at least two bits of communication~\cite{Renner2022}. However, in this work, we introduce a method that supersedes such a constraint.

The goal for Alice is still to prepare the distribution ${\rho(\vec{\lambda})=\Theta(\vec{v}_a\cdot \vec{\lambda})/\pi}$ to Bob. The improvement here comes from the way to achieve that. In the previous approach, Alice chooses her output first (according to the probabilities $p_\pm$) and then samples the corresponding distribution $\Theta(\vec{\lambda}\cdot \vec{v}_\pm)/\pi$. In our approach, Alice samples first the weighted sum $p_+\ \Theta(\vec{\lambda}\cdot \vec{v}_+)/\pi + p_-\ \Theta(\vec{\lambda}\cdot \vec{v}_-)/\pi$ of these two distributions. Afterwards (Step 3), she chooses her output $a=\pm 1$ in such a way that, conditioned on her output $a$, the resulting distribution of $\vec{\lambda}$ becomes exactly $\Theta(\vec{v}_a \cdot \vec{\lambda})/\pi$. At the same time, the weights $p_\pm$ ensure that Alice outputs according to the correct marginals. More formally, all our simulation protocols fit into the following general framework:
\setcounter{protocol}{-1}
\begin{protocol} 
General framework: \label{protocolgenfram}
\begin{enumerate}
    \item Alice chooses her basis $\vec{x}$ and calculates $p_\pm , \vec{v}_\pm$.
    \item Alice and Bob share two (or three) vectors $\vec{\lambda}_i \in S_2$ according to a certain distribution (specified later). Alice informs Bob to choose one of these vectors such that the resulting distribution of the chosen vector $\vec{\lambda}$ becomes:
    \begin{align}
        \rho_{\vec{x}}(\vec{\lambda}):= p_+\ \Theta(\vec{v}_+ \cdot \vec{\lambda})/\pi +p_-\ \Theta(\vec{v}_- \cdot \vec{\lambda})/\pi \, . \label{defrho}
    \end{align}
    \item Given that $\vec{\lambda}$, Alice outputs $a=\pm 1$ with probability 
        \begin{align}
            p_A(a|\vec{x},\vec{\lambda})=\frac{p_a\ \Theta(\vec{\lambda}\cdot \vec{v}_a)/\pi}{\rho_{\vec{x}}(\vec{\lambda})}  \, . \label{aliceoutput}
        \end{align}
    \item Bob chooses his basis $\vec{y}$ and outputs $b=\sgn(\vec{y}\cdot \vec{\lambda})$.
\end{enumerate}
\end{protocol}
\begin{proof}
    To see that this is sufficient to simulate the correct statistics, we first calculate the total probability that Alice outputs $a=\pm 1$ in Step 3:
\begin{align}
    p_A(a|\vec{x})&=\int_{S_2} p_A(a|\vec{x},\vec{\lambda})\cdot \rho_{\vec{x}}(\vec{\lambda}) \, \mathrm{d}\vec{\lambda}\\
    &=\int_{S_2}  p_a\ \Theta(\vec{\lambda}\cdot \vec{v}_\pm)/\pi \, \mathrm{d}\vec{\lambda} =p_a \, . \label{marginals}
\end{align}
For the last step, see Eq.~\eqref{normalization} in App.~\ref{prooflemma}. Now we can show that, given Alice outputs $a=\pm 1$, the conditional distribution of the resulting vector $\vec{\lambda}$ is:
\begin{align}
    \rho_{\vec{x}}(\vec{\lambda}|a)=\frac{p_A(a|\vec{x},\vec{\lambda})\cdot \rho_{\vec{x}}(\vec{\lambda})}{p_A(a|\vec{x})}=\frac{1}{\pi}\  \Theta(\vec{\lambda}\cdot \vec{v}_a) \, .
\end{align}
As in the previous approach, Lemma~\ref{lemma1} ensures that Bob outputs $b$ in Step 4 according to $p(b|\vec{x},\vec{y},a)= |\braket{\vec{v}_a}{b\vec{y}}|^2$. All together, the total probability of this procedure becomes $p_C(a,b|\vec{x},\vec{y})=p_a\cdot p(b|\vec{x},\vec{y},a)=p_a\cdot |\braket{\vec{v}_a}{b\vec{y}}|^2$. This equals $p_Q(a,b|\vec{x},\vec{y})$ as given in Eq.~\eqref{rewrite}.
\end{proof}

Hence, the amount of communication to simulate a qubit pair reduces to an efficient way to sample the distributions $\rho_{\vec{x}}(\vec{\lambda})$. Clearly, the ability to sample each term $\Theta(\vec{\lambda}\cdot \vec{v}_\pm)/\pi$ individually (as in the previous approach~\cite{cerf2000, tonerbacon2003}) implies the possibility to sample the weighted sum of these two terms $\rho_{\vec{x}}(\vec{\lambda})$. However, in general, this is not necessary, and we find more efficient ways to do that. The improvement comes from the fact, that the two post-measurement states are not independent of each other but satisfy the following relation:
\begin{align}
    p_+\ketbra*{\vec{v}_+}+p_-\ketbra*{\vec{v}_-}=\Tr_A[\ketbra*{\Psi_{AB}}] \, .
\end{align}
This follows directly from Eq.~\eqref{vplusminus} and $\ketbra*{+\vec{x}}+\ketbra*{-\vec{x}}=\mathds{1}$. In the Bloch vector representation, this equation becomes:
\begin{align}
    p_+\ \vec{v}_++p_-\ \vec{v}_-=(2p-1)\ \vec{z} \, , \label{nonsignalling}
\end{align}
where we define $\vec{z}:=(0,0,1)^T$. For instance, if the state is local ($p=1$), the two post-measurement states are always $\vec{v}_\pm=\vec{z}$, independent of Alice's measurement $\vec{x}$. In that case, the distributions $\rho_{\vec{x}}(\vec{\lambda})\equiv \Theta(\vec{\lambda}\cdot \vec{z})/\pi$ in Eq.~\eqref{defrho} are constant and do not require any communication to be implemented. If the state is weakly entangled  ($p\lesssim 1$), one post-measurement state $\vec{v}_a$ is still very close to the vector $\vec{z}$. In this way, it turns out that, for every $\vec{x}$, the distribution $\rho_{\vec{x}}(\vec{\lambda})$ is dominated by a constant part proportional to $\Theta(\vec{\lambda}\cdot \vec{z})/\pi$. More formally, we can define
\begin{align}
    \tilde{\rho}_{\vec{x}}(\vec{\lambda}):=&\rho_{\vec{x}}(\vec{\lambda})-\frac{(2p-1)}{\pi}\ \Theta(\vec{\lambda}\cdot \vec{z}) \, .
\end{align}
In Appendix~\ref{rhotilde}, we prove the following properties of that distribution and give an illustration of them (see Fig.~\ref{fig2}).

\begin{lemma}\label{lemma2}
The distribution $\tilde{\rho}_{\vec{x}}(\vec{\lambda})$ defined above is positive, $\tilde{\rho}_{\vec{x}}(\vec{\lambda})\geq 0$ and  sub-normalized, $\int_{S_2} \tilde{\rho}_{\vec{x}}(\vec{\lambda}) \ \mathrm{d}\vec{\lambda}=2(1-p)$. Additionally, it respects the two upper bounds, $\tilde{\rho}_{\vec{x}}(\vec{\lambda})\leq \frac{\sqrt{p(1-p)}}{\pi}$  and $\tilde{\rho}_{\vec{x}}(\vec{\lambda})\leq \frac{p_\pm}{\pi} |\vec{\lambda}\cdot \vec{v}_\pm|$.
\end{lemma}

\section{One bit protocol for weakly entangled states} 
In particular, when the state is weakly entangled, the extra term $\tilde{\rho}_{\vec{x}}(\vec{\lambda})$ remains small. This allows us to find the
following protocol for the range $p\geq 1/2+\sqrt{3}/4 \approx 0.933$.
\begin{protocol}[$1/2+\sqrt{3}/4\leq p\leq 1$, 1 bit] \label{protocol1}
Same as Protocol~\ref{protocolgenfram} with the following 2. Step:\\
Alice and Bob share two normalized three-dimensional vectors $\vec{\lambda}_1, \vec{\lambda}_2 \in S_2$ according to the distribution:
    \begin{align}
    \rho(\vec{\lambda}_1)=\frac{1}{4\pi}\, ,&&
    \rho(\vec{\lambda}_2)=\frac{1}{\pi}\ \Theta(\vec{\lambda}_2 \cdot \vec{z}) \, .
\end{align}
    Alice sets $c=1$ with probability:
    \begin{align}
    \begin{split}
        p_A(c=1|\vec{x},\vec{\lambda}_1)=(4\pi)\cdot \tilde{\rho}_{\vec{x}} (\vec{\lambda}_1)
    \end{split}
    \end{align}
    and otherwise she sets $c=2$. She communicates the bit $c$ to Bob. Both set $\vec{\lambda}:=\vec{\lambda}_c$ and reject the other vector.
\end{protocol}
\begin{proof}
Whenever Alice chooses the first vector, the resulting distribution of the chosen vector is precisely $p_A(c=1|\vec{x},\vec{\lambda}_1)\cdot \rho(\vec{\lambda}_1)=\tilde{\rho}_{\vec{x}}(\vec{\lambda}_1)$. The total probability that she chooses the first vector is $\int_{S_2} p_A(c=1|\vec{x},\vec{\lambda}_1)\cdot \rho(\vec{\lambda}_1) \, \mathrm{d}\vec{\lambda}_1=\int_{S_2} \tilde{\rho}_{\vec{x}}(\vec{\lambda}_1) \, \mathrm{d}\vec{\lambda}_1 = 2(1-p)$ (see Lemma~\ref{lemma2}). In all the remaining cases, with total probability $2p-1$, she chooses vector $\vec{\lambda}_2$, distributed as $\Theta(\vec{\lambda}_2 \cdot \vec{z})/\pi$. Therefore, the total distribution of the chosen vector $\vec{\lambda}:=\vec{\lambda}_c$ becomes the desired distribution
\begin{align}
    \tilde{\rho}_{\vec{x}}(\vec{\lambda})+\frac{(2p-1)}{\pi}\ \Theta(\vec{\lambda}\cdot \vec{z})=\rho_{\vec{x}}(\vec{\lambda}) \, .
\end{align}
In order for the protocol to be well defined, it has to hold that $0\leq p(c=1|\vec{x},\vec{\lambda}_1)\leq 1$, hence $0\leq \tilde{\rho}_{\vec{x}} (\vec{\lambda}_1)\leq 1/(4\pi)$. As a consequence of Lemma~\ref{lemma2}, this is true whenever $1/2+\sqrt{3}/4 \leq p\leq 1$.
\end{proof}

\begin{figure}
    \centering
    \includegraphics[width=0.45 \textwidth]{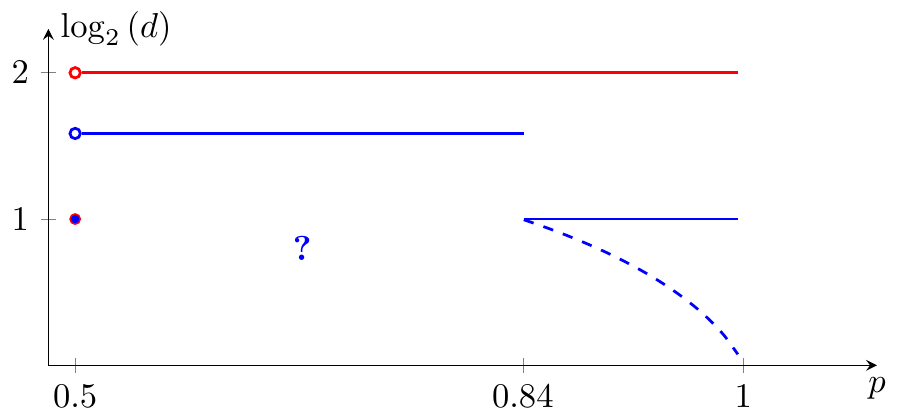}
\caption{
Length of the classical message $d$ required to simulate a qubit pair ${\ket{\Psi_{AB}}=\sqrt{p}\ket{00}+\sqrt{1-p}\ket{11}}$ as a function in $p$. The previous best result, from Toner and Bacon~\cite{tonerbacon2003}, is presented in red. Our novel results are presented in blue. The dashed curve in blue represents the fraction of rounds where Alice needs to send a bit to Bob. The main open question is whether a single bit is sufficient for simulating qubit pairs with $1/2<p<0.84$.
}
    \label{figcomparison}
\end{figure}
Clearly, the simulation of weakly entangled states requires some communication since all pure entangled quantum states violate a Bell inequality \cite{Gisin1991}. It is surprising that the minimal amount of information (1 bit) is already sufficient to reproduce the correlations for all projective measurements. However, we can even improve that protocol. In Appendix~\ref{improvedprot} (Protocol~\ref{improvedonebitprotocol}), we show how to simulate every weakly entangled state with $0.835 \leq p \leq 1$ by communicating only a single bit. Moreover, it turns out that this bit is not necessary in each round. In fact, Alice sends the bit in only a ratio of $N(p)$ of the rounds, where
\begin{align}
    N(p):=\frac{2p(1-p)}{2p-1} \log{\left(\frac{p}{1-p}\right)}+2(1-p) \, .
\end{align}
In the remaining rounds, with probability $1-N(p)$, they do not communicate with each other. In the limit where $p$ approaches one, the function $N(p)$ approaches zero. Hence, if the state is very weakly entangled, a perfect simulation is possible even though they communicate a single bit only in a small fraction of rounds (see dashed curve in Fig.~\ref{figcomparison}). It is known that a simulation of a maximally entangled state without communication in some fraction of rounds is impossible. This would contradict the fact that the singlet has no local part \cite{elitzur1992, barrett2006}. Hence, our result shows that simulating weakly entangled states requires strictly fewer communication resources than simulating a maximally entangled one. Interestingly, we can also use our approach to quantify the local content of any pure entangled two-qubit state (see Appendix~\ref{applocalcontent}). More precisely, we maximize the fraction of rounds in which no communication is necessary. This provides an independent proof of a result by Portmann et al.~\cite{Portmann2012localpart}.

\section{Trit protocol for arbitrary entangled pairs}
It is worth mentioning, that we also recover a one-bit protocol for simulating the maximally entangled state ($p=1/2$) in our framework. In that case, there is another geometric argument that allows to sample the distributions $\rho_{\vec{x}}(\vec{\lambda})$ efficiently. More precisely, the two post-measurement states are always opposite of each other $\vec{v}_-=-\vec{v}_+$ and it holds that $p_+=p_-=1/2$. In this way, the distribution $\rho_{\vec{x}}(\vec{\lambda})$ turns out to be $\rho_{\vec{x}}(\vec{\lambda})=|\vec{\lambda}\cdot \vec{v}_+|/(2\pi)$. It was already observed, by Degorre et al.~\cite{degorre2005}, that this distribution can be sampled by communicating only a single bit of communication (see Appendix~\ref{appendixb} for details and an independent proof). Here, we connect this with the techniques developed for Protocol~~\ref{protocol1} to present a protocol that simulates all entangled qubit pairs by communicating a classical trit.

\begin{protocol}[$1/2\leq p\leq 1$, 1 trit]
Same as Protocol~\ref{protocolgenfram} with the following 2. Step:\\
Alice and Bob share three normalized three-dimensional vectors $\vec{\lambda}_1,\vec{\lambda}_2,\vec{\lambda}_3 \in S_2$ according to the following distribution:
\begin{align}
    \rho(\vec{\lambda}_1)=\frac{1}{4\pi}\, ,&&
    \rho(\vec{\lambda}_2)=\frac{1}{4\pi}\, ,&&
    \rho(\vec{\lambda}_3)=\frac{1}{\pi}\ \Theta(\vec{\lambda}_3 \cdot \vec{z})\, .
\end{align}
If $p_+\leq 0.5$ she sets $\vec{v}:=\vec{v}_+$, otherwise she sets $\vec{v}:=\vec{v}_-$. Afterwards, she sets $c=1$ if $|\vec{v}\cdot \vec{\lambda}_1|\geq |\vec{v}\cdot \vec{\lambda}_2|$ and $c=2$ otherwise. Finally, with probability
\begin{align}
    \begin{split}
    p_A(t=c|\vec{x},\vec{\lambda}_c)=\frac{\tilde{\rho}_{\vec{x}}(\vec{\lambda}_c)}{\frac{1}{2\pi} |\vec{\lambda}_c \cdot \vec{v}|}
    \end{split}
\end{align}
she sets $t=c$ and otherwise, she sets $t=3$. She communicates the trit $t$ to Bob. Both set $\vec{\lambda}:=\vec{\lambda}_t$ and reject the other two vectors.
\end{protocol}

\begin{proof}
We show that the distribution of the shared vector $\vec{\lambda}$ becomes exactly the required $\rho_{\vec{x}}(\vec{\lambda})$. Consider the step before Alice sets $t=c$ or $t=3$. As a result of Ref.~\cite{degorre2005} (see Protocol~\ref{protocoldegorre} in Appendix~\ref{appendixb} for details), the distribution of the vector $\vec{\lambda}_c$ is $\rho(\vec{\lambda}_c)=\frac{1}{2\pi}|\vec{\lambda}_c\cdot \vec{v}|$. Now we use a similar idea as in the protocol for weakly entangled states. Whenever she sets $t=c$, the resulting distribution of the chosen vector is precisely $p_A(t=c|\vec{x},\vec{\lambda}_c)\cdot \rho(\vec{\lambda}_c)=\tilde{\rho}_{\vec{x}}(\vec{\lambda}_c)$. The total probability that she sets $t=c$ is $\int_{S_2} p_A(t=c|\vec{x},\vec{\lambda}_c)\cdot \rho(\vec{\lambda}_c) \, \mathrm{d}\vec{\lambda}_c=\int_{S_2} \tilde{\rho}_{\vec{x}}(\vec{\lambda}_c) \, \mathrm{d}\vec{\lambda}_c = 2(1-p)$ (see Lemma~\ref{lemma2}). In all the remaining cases, with total probability $2p-1$, she chooses vector $\vec{\lambda}_3$, distributed as $\Theta(\vec{\lambda}_3 \cdot \vec{z})/\pi$. Therefore, the total distribution of the chosen vector $\vec{\lambda}:=\vec{\lambda}_t$ becomes the desired distribution
\begin{align}
    \tilde{\rho}_{\vec{x}}(\vec{\lambda})+\frac{(2p-1)}{\pi}\ \Theta(\vec{\lambda}\cdot \vec{z})=\rho_{\vec{x}}(\vec{\lambda}) \, .
\end{align} 
The fact that $0\leq p_A(t=c|\vec{x},\vec{\lambda}_c) \leq 1$ follows from $0\leq \tilde{\rho}_{\vec{x}}(\vec{\lambda})\leq \frac{p_\pm}{\pi} |\vec{\lambda}\cdot \vec{v}_\pm|$ in Lemma~\ref{lemma2}.
\end{proof}

\section{Discussion} 
To conclude, we showed that a classical trit is enough for simulating the outcomes of local projective measurements on any entangled qubit pair. For weakly entangled states, we proved that already a single bit is sufficient. In the latter case, Alice does not need to send the bit in all the rounds, which is impossible for a maximally entangled state \cite{elitzur1992, barrett2006}. In this way, we show that simulating weakly entangled states is strictly simpler than simulating maximally entangled ones.

The main open question now is whether a single bit is sufficient to simulate every entangled qubit pair, see Fig.~\ref{figcomparison}. Recently, numerical evidence has been reported by Sidajaya et al.~\cite{sidajaya23} that a single bit is indeed enough. However, an analytical model is still missing. We remark that our framework is in principle capable of providing such a model. The challenge becomes to find, for each qubit pair, a distribution of two shared random vectors, 
such that Alice can sample $\rho_{\vec{x}}(\vec{\lambda})$ for every measurement basis $\vec{x}$. 
In all protocols considered here, the shared vectors are independent of each other, i.e., $\rho(\vec{\lambda}_1,\vec{\lambda}_2)=\rho(\vec{\lambda}_1)\cdot \rho(\vec{\lambda}_2)$. Dropping this constraint may be a way to extend our approach and may lead to a complete solution to this longstanding open question \cite{Gisinreview2018,Brassard2003, Brunner_2014}. 


\begin{acknowledgments}
We thank \v{C}aslav Brukner, Valerio Scarani,  Peter Sidajaya, Armin Tavakoli, Isadora Veeren, and Bai Chu Yu for fruitful discussions. Furthermore, we thank Nicolas Brunner and Nicolas Gisin for pointing out Ref.~\cite{Portmann2012localpart} to us.

This research was funded in whole, or in part, by the Austrian
Science Fund (FWF) through BeyondC (F7103). For the purpose of open access, the author has applied a CC BY public copyright license to any Author Accepted Manuscript version arising from this submission.
%

This project has received funding from the European Union’s Horizon 2020 research and innovation programme under the Marie Skłodowska-Curie grant agreement No 801110. It reflects only the authors' view, the EU Agency is not responsible for any use that may be made of the information it contains. ESQ has received funding from the Austrian Federal Ministry of Education, Science and Research (BMBWF).
\end{acknowledgments}


\nocite{apsrev42Control} 
\bibliographystyle{0_MTQ_apsrev4-2_corrected}
\bibliography{bib.bib}


\onecolumngrid

\appendix

\section{Proof of Lemma~\ref{lemma1}} \label{prooflemma}
\setcounter{lemma}{0}
\begin{lemma}
Bob receives a vector $\vec{\lambda}\in S_2$ distributed as $\rho(\vec{\lambda})=\Theta(\vec{v}\cdot \vec{\lambda})/\pi$ and outputs $b=\sgn(\vec{y}\cdot \vec{\lambda})$. For every qubit state $\vec{v}\in S_2$ and measurement $\vec{y}\in S_2$ this reproduces quantum correlations:
\begin{align}
    p(b=\pm 1|\vec{y},\vec{v})&=(1 \pm \vec{y}\cdot \vec{v})/2=|\braket{\pm \vec{y}}{\vec{v}}|^2\, .
\end{align}
\end{lemma}
\begin{proof}
Bob outputs $b=+1$ if and only if $\vec{y}\cdot \vec{\lambda}\geq 0$. Therefore, the total probability that Bob outputs $b=+1$ becomes:
\begin{align}
    p(b=+1| \vec{y}, \vec{v})=\int_{S_2} H(\vec{y}\cdot \vec{\lambda})\cdot \ \rho(\vec{\lambda})\  \mathrm{d}\vec{\lambda}=\frac{1}{\pi}\int_{S_2} H(\vec{y}\cdot \vec{\lambda})\cdot \ \Theta(\vec{v}\cdot \vec{\lambda})\  \mathrm{d}\vec{\lambda} \, .
\end{align}
Here, $H(z)$ is the Heaviside function ($H(z)=1$ if $z\geq 0$ and $H(z)=0$ if $z< 0$) and $\Theta(z):=H(z)\cdot z$. The evaluation of the exact same integral is done in Ref.~\cite{Renner2022} (Lemma~1) and in similar forms also in Ref.~\cite{gisingisin1999, cerf2000, degorre2005}. For the sake of completeness, we restate the same proof as in Ref.~\cite{Renner2022} here:\\
"Note that both functions in the integral $H(\vec{y}\cdot \vec{\lambda})$ and $\Theta(\vec{v}\cdot \vec{\lambda})$ have support in only one half of the total sphere (the hemisphere centered around $\vec{y}$ and $\vec{v}$, respectively). For example, if $\vec{v}=-\vec{y}$ these two hemispheres are exactly opposite of each other and the integral becomes zero. For all other cases, we can observe that the value of the integral depends only on the angle between $\vec{y}$ and $\vec{v}$, because the whole expression is spherically symmetric. Therefore, it is enough to evaluate the integral for $\vec{y}=(0,1,0)^T$ and $\vec{v}=(-\sin{\beta},\cos{\beta},0)^T$, where we can choose without loss of generality $0\leq\beta\leq \pi$. Furthermore, we can use spherical coordinates for $\vec{\lambda}=(\sin{\theta}\cdot \cos{\phi}, \sin{\theta}\cdot \sin{\phi}, \cos{\theta})$ (note that $\vert\vec{\lambda}\vert =1$). With this choice of coordinates, the region in which both factors have non-zero support becomes exactly $\beta\leq \phi \leq \pi$ (at the same time, $\theta$ is unrestricted, $0 \leq \theta \leq \pi$). More precisely, $0\leq \phi \leq \pi$ is the support for $H(\vec{y}\cdot \vec{\lambda})$ and $\beta\leq \phi \leq \pi + \beta$ is the support for $\Theta(\vec{v}\cdot \vec{\lambda}$). In this way, the integral becomes:
\begin{align}
    \frac{1}{\pi}\int^{2\pi}_{0} \int^{\pi}_{0} H(\vec{y}\cdot \vec{\lambda})\cdot \ \Theta(\vec{v}\cdot \vec{\lambda}) \cdot \sin{\theta}\  \mathrm{d}\theta  \ \mathrm{d}\phi &=\frac{1}{\pi}\int^{\pi}_{\beta} \int^{\pi}_{0} \ \sin{\phi} \cdot \sin^2{\theta}\  \mathrm{d}\theta  \ \mathrm{d}\phi\\
    &=\frac{1}{2}(1+\cos{\beta})=\frac{1}{2}(1+\vec{y}\cdot \vec{v}) \, ."
\end{align}
Hence, $p(b=+1|\vec{y},\vec{v})=(1+\vec{y}\cdot \vec{v})/2$. Clearly, $p(b=-1| \vec{y}, \vec{v})=1-p(b=+1| \vec{y}, \vec{v})=(1-\vec{y}\cdot \vec{v})/2$.
\end{proof}

It appears several times in this work that $\int_{S_2} \Theta(\vec{\lambda}\cdot \vec{v})/\pi \ \mathrm{d}\vec{\lambda}=1$ for every normalized vector $\vec{v}\in S_2$. The proof follows by a similar calculation as in the above Lemma:
\begin{align}
    \frac{1}{\pi}\int_{S_2} \ \Theta(\vec{\lambda}\cdot \vec{v}) \ \mathrm{d}\vec{\lambda}=\frac{1}{\pi}\int_{S_2} \ H(\vec{\lambda}\cdot \vec{v})\cdot \Theta(\vec{\lambda}\cdot \vec{v}) \ \mathrm{d}\vec{\lambda}=\frac{1}{2}(1+\vec{v}\cdot \vec{v})=1 \, . \label{normalization}
\end{align}
The introduction of the Heaviside function in the second step clearly does not change the integral since $H(\vec{\lambda}\cdot \vec{v})$ has the same support as $\Theta(\vec{\lambda}\cdot \vec{v})$ or, more formally, $\Theta(\vec{\lambda}\cdot \vec{v}):=H(\vec{\lambda}\cdot \vec{v})\cdot (\vec{\lambda}\cdot \vec{v})=H(\vec{\lambda}\cdot \vec{v})^2\cdot (\vec{\lambda}\cdot \vec{v})=H(\vec{\lambda}\cdot \vec{v})\cdot \Theta(\vec{\lambda}\cdot \vec{v})$, where we used that $H(z)=H(z)^2$ for every $z\in \mathbb{R}$.

\section{Protocol for the maximally entangled qubit pair}\label{appendixb}
As mentioned in the main text, in the case of a maximally entangled state ($p=1/2$), there is a similar geometric argument that allows sampling the distributions $\rho_{\vec{x}}(\vec{\lambda})$ efficiently. More precisely, the two post-measurement states are always opposite of each other $\vec{v}_-=-\vec{v}_+$ and it holds that $p_+=p_-=1/2$. Therefore, the distribution $\rho_{\vec{x}}(\vec{\lambda})$ has, for every choice of Alice's measurement $\vec{x}$, the form (note that $\Theta(z)+\Theta(-z)=|z|$ for every $z\in \mathbb{R}$)
\begin{align}
    \rho_{\vec{x}}(\vec{\lambda})=\frac{1}{2\pi}\left(\Theta(\vec{\lambda}\cdot \vec{v}_+)+\ \Theta(\vec{\lambda}\cdot (-\vec{v}_+))\right)=\frac{1}{2\pi}|\vec{\lambda}\cdot \vec{v}_+| \, . \label{distribsinglet}
\end{align}
It was already observed by Degorre et al.~\cite{degorre2005} ("Theorem~6 (The “choice” method)"), that this distribution can be sampled by communicating only a single bit. This leads directly to a simulation of the maximally entangled state with one bit of communication. This is exactly the version of Degorre et al.~\cite{degorre2005} for the protocol of Toner and Bacon for the singlet (see also "Theorem 10 (Communication)" of Ref.~\cite{degorre2005}):

\begin{protocol}[$p=1/2$, 1 bit, from Ref.~\cite{degorre2005}] \label{protocoldegorre}
Same as Protocol~\ref{protocolgenfram} with the following 2. Step:\\
Alice and Bob share two normalized three-dimensional vectors $\vec{\lambda}_1,\vec{\lambda}_2 \in S_2$ that are independent and uniformly distributed on the unit sphere, $\rho(\vec{\lambda}_1)=\rho(\vec{\lambda}_2)=1/4\pi$. Alice sets $c=1$ if $|\vec{v}_+\cdot \vec{\lambda}_1|\geq |\vec{v}_+\cdot \vec{\lambda}_2|$ and $c=2$ otherwise. She communicates the bit $c$ to Bob and both set $\vec{\lambda}:=\vec{\lambda}_c$.\\

\end{protocol}
\begin{proof}
The proof can be found in Ref.~\cite{degorre2005}  ("Theorem~6 (The “choice” method)"). However, we want to give an independent proof here. We can focus first on the case where $\vec{v}_+=\vec{z}$. All other cases are analogous due to the spherical symmetry of the problem. In that case, we can write $\vec{\lambda}_1$ and $\vec{\lambda}_2$ in spherical coordinates, $\vec{\lambda}_i=(\sin{\theta_i}\cdot \cos{\phi_i}, \sin{\theta_i}\cdot \sin{\phi_i}, \cos{\theta_i})$. In this notation, Alice picks $\vec{\lambda}_1$ if and only if $|\vec{\lambda}_1\cdot \vec{z}|=|\cos{\theta_1}|\geq |\vec{\lambda}_2\cdot \vec{z}|=|\cos{\theta_2}|$. For a given $\vec{\lambda}_1$, this happens with probability
    \begin{align}
        \int_{S_2}\ H(|\vec{\lambda}_1\cdot \vec{z}|-|\vec{\lambda}_2\cdot \vec{z}|)\cdot \rho(\vec{\lambda}_2) \ \mathrm{d}\vec{\lambda}_2=&\frac{1}{4\pi}\ \int^{2\pi}_{0} \int^{\pi}_{0} H(|\cos{\theta_1}|-|\cos{\theta_2}|) \cdot \sin{\theta_2}\  \mathrm{d}\theta_2  \ \mathrm{d}\phi_2\\
        =&\frac{1}{2}\  \int^{\pi}_{0} H(|\cos{\theta_1}|-|\cos{\theta_2}|) \cdot \sin{\theta_2}\  \mathrm{d}\theta_2
    \end{align}
If $0\leq \theta_1\leq \pi/2$ and hence $\cos{\theta_1}\geq 0$, the region where $H(|\cos{\theta_1}|-|\cos{\theta_2}|)=1$ becomes exactly $\theta_1\leq \theta_2 \leq \pi- \theta_1$ and hence the above integral becomes:
\begin{align}
        \frac{1}{2}\  \int^{\pi-\theta_1}_{\theta_1} \sin{\theta_2}\  \mathrm{d}\theta_2
        =\frac{1}{2}\  \left[ -\cos{\theta_2} \right]^{\pi-\theta_1}_{\theta_1}=\left( -\cos{(\pi-\theta_1)}+\cos{(\theta_1)} \right)/2=\cos{\theta_1}=|\vec{z}\cdot \vec{\lambda}_1| \, .
\end{align}
For $\pi/2\leq \theta_1\leq \pi$ and hence $\cos{\theta_1}\leq 0$ a similar calculation leads to $-\cos{\theta_1}=|\cos{\theta_1}|=|\vec{z}\cdot \vec{\lambda}_1|$. (One can also observe that the integral depends only on $|\cos{\theta_1}|$, which leads to the same statement.) Hence, whenever Alice chooses $\vec{\lambda}:=\vec{\lambda}_1$ the distribution of that vector becomes $\rho(\vec{\lambda}_1)\cdot|\vec{z}\cdot \vec{\lambda}_1|=|\vec{z}\cdot \vec{\lambda}_1|/(4\pi)$. Analogously, whenever Alice chooses $\vec{\lambda}:=\vec{\lambda}_2$ the distribution of that chosen vector is again $|\vec{z}\cdot \vec{\lambda}_2|/(4\pi)$, due to the symmetric roles of $\vec{\lambda}_1$ and $\vec{\lambda}_2$. Hence, the distribution of the chosen vector $\vec{\lambda}$ becomes, in total, the sum of these two terms $\rho(\vec{\lambda})=|\vec{z}\cdot \vec{\lambda}|/(2\pi)$. For a general vector $\vec{v}_+$, the analog expression $\rho(\vec{\lambda})=|\vec{v}_+\cdot \vec{\lambda}|/(2\pi)$ holds, because of the spherical symmetry of the protocol.
\end{proof}

We also want to remark here, that in the case of a maximally entangled state, Alice's response function in the third step Eq.~\eqref{aliceoutput} can be, due to Eq.~\eqref{distribsinglet}, rewritten into $p_A(a=\pm 1|\vec{x},\vec{\lambda})=H(\vec{\lambda}\cdot \vec{v}_+)$, or, equivalently, $a=\sgn(\vec{\lambda}\cdot \vec{v}_+)$.

\subsection{"Classical teleportation" protocol}
With this observation, we can also understand the classical teleportation protocol from the main text. To avoid confusion, this protocol is \emph{not} of the form given in Protocol~\ref{protocolgenfram}:\\

\begin{protocol}\label{protocol4}
The following protocol simulates a qubit in a prepare-and-measure scenario:
\begin{enumerate}
    \item Alice chooses the quantum state $\ketbra{v}=(\mathds{1}+\vec{v}\cdot \vec{\sigma})/2$ she wants to send to Bob.
    \item Alice and Bob share two normalized three-dimensional vectors $\vec{\lambda}_1,\vec{\lambda}_2 \in S_2$ that are independent and uniformly distributed on the unit sphere, $\rho(\vec{\lambda}_1)=\rho(\vec{\lambda}_2)=1/4\pi$. Alice sets $c_1=1$ if $|\vec{v}\cdot \vec{\lambda}_1|\geq |\vec{v}\cdot \vec{\lambda}_2|$ and $c_1=2$ otherwise. In addition, Alice defines a second bit $c_2=\sgn(\vec{\lambda}_{c_1}\cdot \vec{v})$. She communicates the two bits $c_1$ and $c_2$ to Bob and both set $\vec{\lambda}:=c_2\ \vec{\lambda}_{c_1}$.
    \item Bob outputs $b=\sgn(\vec{y}\cdot \vec{\lambda})$.
\end{enumerate}
\end{protocol}
\begin{proof}
    As a result of the above (Protocol~\ref{protocoldegorre}), the distribution of the vector $\vec{\lambda}_{c_1}$ is $\rho(\vec{\lambda}_{c_1})=|\vec{\lambda}_{c_1}\cdot \vec{v}|/(2\pi)$. When he defines  $\vec{\lambda}:=c_2\ \vec{\lambda}_{c_1}$, he exactly flips the vector $\vec{\lambda}_{c_1}$ if and only if $\vec{\lambda}_{c_1}\cdot \vec{v}< 0$. With the additional flip, he obtains the distribution:
\begin{align}
    \frac{1}{2\pi}|\vec{\lambda}_{c_1}\cdot \vec{v}|=\frac{1}{2\pi}\left(\Theta(\vec{\lambda}_{c_1}\cdot \vec{v})+\ \Theta(\vec{\lambda}_{c_1}\cdot (-\vec{v}))\right)
    \xlongrightarrow{flip}\frac{1}{2\pi}\left(\Theta(\vec{\lambda}\cdot \vec{v})+\ \Theta((-\vec{\lambda})\cdot (-\vec{v}))\right)=\frac{1}{\pi}\Theta(\vec{\lambda}\cdot \vec{v}) \, .
\end{align}
Therefore, the distribution of the vector $\vec{\lambda}$ becomes $\Theta(\vec{v} \cdot \vec{\lambda})/\pi$. In this way, Alice managed to send exactly a classical description of the state $\ketbra*{\vec{v}}$ to Bob. More precisely, Lemma~\ref{lemma1} ensures that Bob outputs according to $p(b=\pm 1|\vec{v},\vec{y})=(1\pm \vec{y}\cdot \vec{v})/2$, as required by quantum mechanics.
\end{proof}

Note that, in the main text, the second step is formulated as follows: "Alice and Bob share four normalized three-dimensional vectors $\vec{\lambda}_1,\vec{\lambda}_2, \vec{\lambda}_3, \vec{\lambda}_4\in S_2$. The first two $\vec{\lambda}_1$ and $\vec{\lambda}_2$ are uniformly and independently distributed on the sphere, whereas $\vec{\lambda}_3=-\vec{\lambda}_1$ and $\vec{\lambda}_4=-\vec{\lambda}_2$. From these four vectors, Alice chooses the one that maximizes $\vec{\lambda}_i\cdot \vec{v}$ and communicates the result to Bob and both set $\vec{\lambda}:=\vec{\lambda}_i$."

This is just a reformulation of the second step in Protocol~\ref{protocol4} and both versions are equivalent. To see this, fix $\vec{\lambda}_1$ and $\vec{\lambda}_2$. If $|\vec{v}\cdot \vec{\lambda}_1|\geq |\vec{v}\cdot \vec{\lambda}_2|$ and $\vec{v}\cdot \vec{\lambda}_1\geq 0$, Alice will send $c_1=1$ and $c_2=+1$ and both set $\vec{\lambda}:=c_2\ \vec{\lambda}_{c_1}=\vec{\lambda}_1$ in step two of Protocol~\ref{protocol4}. In the reformulation, it turns out that the vector that maximizes $\vec{v}\cdot \vec{\lambda}_i$ is precisely $\vec{\lambda}_1$ since $|\vec{v}\cdot \vec{\lambda}_1|\geq |\vec{v}\cdot \vec{\lambda}_2|$ and $\vec{v}\cdot \vec{\lambda}_1\geq 0$ imply that $\vec{v}\cdot \vec{\lambda}_1\geq \vec{v}\cdot \vec{\lambda}_2$; $\vec{v}\cdot \vec{\lambda}_1\geq \vec{v}\cdot \vec{\lambda}_3=-\vec{v}\cdot \vec{\lambda}_1$ as well as $\vec{v}\cdot \vec{\lambda}_1\geq \vec{v}\cdot \vec{\lambda}_4=-\vec{v}\cdot \vec{\lambda}_2$. With similar arguments, one can check that, for a fixed $\vec{\lambda}_1$ and $\vec{\lambda}_2$, they always choose the same vector $\vec{\lambda}$ in both versions.

\section{Properties of $\tilde{\rho}_{\vec{x}}(\vec{\lambda})$} \label{rhotilde}

\begin{figure}[hbt]
    \centering
    \includegraphics[width=0.95 \textwidth]{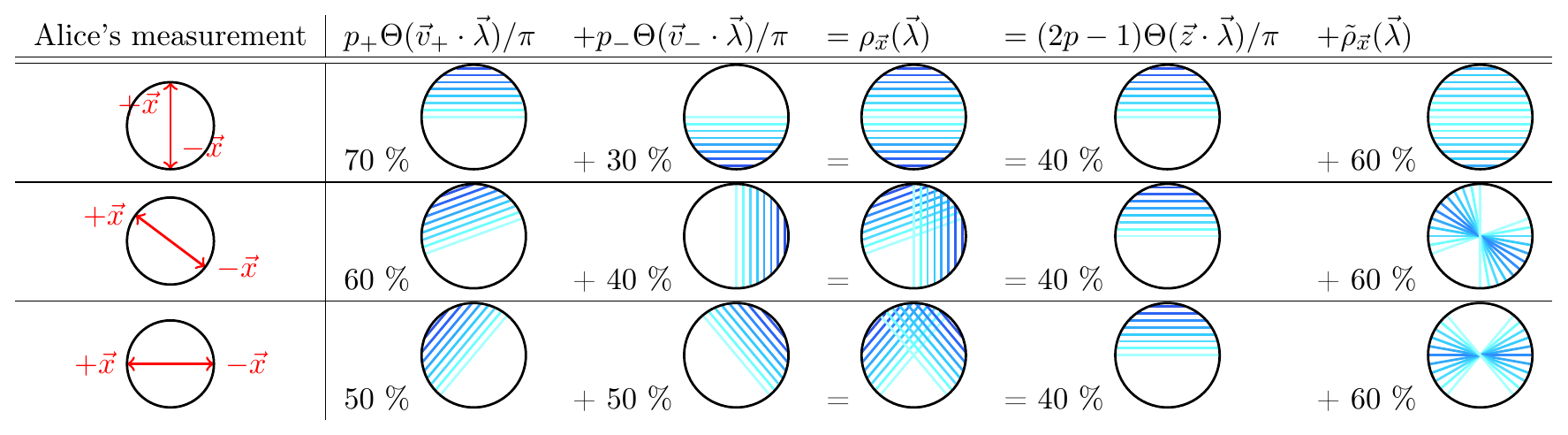}
    \caption{A sketch of the relevant distributions for the state $\ket{\Psi_{AB}}=\sqrt{0.7}\ket{00}+\sqrt{0.3}\ket{11}$: In the previous approach, Alice tosses a biased coin according to the marginals of her measurement. Then she uses the classical teleportation protocol to create on Bob's side a vector according to the distribution $\Theta (\vec{v}_a\cdot \vec{\lambda})/\pi$, from which Bob can reproduce quantum statistics (see Lemma~\ref{lemma1}). However, it turns out that it is enough to sample only the sum of these two distributions $\rho_{\vec{x}}(\vec{\lambda})$. Since all of these distributions (for every $\vec{x}$) can be rewritten into a constant part and the extra term $\tilde{\rho}_{\vec{x}}(\vec{\lambda})$, one can sample all of these distributions by communicating only a classical trit. If the state is very weakly entangled ($p\to 1$), it turns out that the constant part with weight $2p-1$ dominates and already one bit is sufficient to sample all the distributions $\rho_{\vec{x}}(\vec{\lambda})$.}
    \label{fig2}
\end{figure}

In this section, we prove several properties of the distribution $\tilde{\rho}_{\vec{x}}(\vec{\lambda})$ that are crucial for our protocols. Let us recall that,
\begin{align}
    \tilde{\rho}_{\vec{x}}(\vec{\lambda}):=&\frac{1}{\pi}\left( p_+\ \Theta(\vec{\lambda}\cdot \vec{v}_+)+p_- \ \Theta(\vec{\lambda}\cdot \vec{v}_-)-(2p-1)\ \Theta(\vec{\lambda}\cdot \vec{z})\right)
\end{align}
where, $0\leq p_\pm \leq 1$, $p_++p_-=1$ and $p\geq0.5$ (therefore $0\leq (2p-1)\leq 1$) and all vectors are normalized vectors on the Bloch sphere $\vec{\lambda},\vec{v}_+,\vec{v}_-,\vec{z}\in S_2$. The only relevant equation that we need for the proof is Eq.~\eqref{nonsignalling}, which reads as 
\begin{equation}
    p_+\ \vec{v}_++p_-\ \vec{v}_-=(2p-1)\ \vec{z},
\end{equation} and the definition of $\Theta(z)$:
\begin{align}
\Theta(z):=\biggl\lbrace \begin{array}{cc}
       z & \text{if } z\geq 0 \\
        0 & \text{if } z<0
    \end{array}\, .
\end{align}

\begin{lemma}
The Distribution $\tilde{\rho}_{\vec{x}}(\vec{\lambda})$ defined above satisfies the following properties:
\begin{enumerate}[label=(\roman*)]
    \item positive: $\tilde{\rho}_{\vec{x}}(\vec{\lambda})\geq 0$
    \item symmetric: $\tilde{\rho}_{\vec{x}}(\vec{\lambda})=\tilde{\rho}_{\vec{x}}(-\vec{\lambda})$
    \item area: $\int_{S_2} \tilde{\rho}_{\vec{x}}(\vec{\lambda}) \ \mathrm{d}\vec{\lambda}=2(1-p)$
    \item 1st bound: $\tilde{\rho}_{\vec{x}}(\vec{\lambda})\leq \frac{p_\pm}{\pi} |\vec{\lambda}\cdot \vec{v}_\pm|$
    \item 2nd bound: $\tilde{\rho}_{\vec{x}}(\vec{\lambda})\leq \frac{1}{2\pi}\frac{1-C^2}{\sqrt{1-C^2\sin^2{(\theta)}}+C|\cos{(\theta)}|}$ with $C:=2p-1$ and $\cos{(\theta)}=\vec{\lambda}\cdot \vec{z}$
    \item 3rd bound: $\tilde{\rho}_{\vec{x}}(\vec{\lambda})\leq \frac{\sqrt{p(1-p)}}{\pi}$
\end{enumerate}
\end{lemma}
\begin{proof}
Most of these properties follow directly from the fact that the function $\Theta(a)$ is convex and satisfies:
\begin{align}
    \forall a,b\in \mathbb{R}:\ \Theta(a)+\Theta(b)\geq \Theta(a+b) \, .
\end{align}
Furthermore, we use the following property frequently:
\begin{align}
    \forall a,b\in \mathbb{R} \text{ with } a\geq 0:\ \Theta(a\cdot b)= a\cdot \Theta(b) \, .
\end{align}
Furthermore, $\Theta(a)+\Theta(-a)=|a|$ as well as $\Theta(a)-\Theta(-a)=a$ for all $ a\in \mathbb{R}$. All of these properties follow directly from the definition of $\Theta(a)$.\\

\textit{(i) positive:}\\
We can use $\Theta(a)+\Theta(b)\geq \Theta(a+b)$ with $a= p_+\ \vec{\lambda}\cdot \vec{v}_+$ and $b= p_-\ \vec{\lambda}\cdot \vec{v}_-$. As a consequence of $p_+\ \vec{v}_++p_-\ \vec{v}_-=(2p-1)\ \vec{z}$ we obtain $a+b=p_+\ \vec{\lambda}\cdot \vec{v}_++p_-\ \vec{\lambda}\cdot \vec{v}_-=\vec{\lambda}\cdot(p_+\ \vec{v}_++p_-\ \vec{v}_-)=(2p-1)\ \vec{\lambda}\cdot\vec{z}$ and therefore:
\begin{align}
    \Theta(p_+\ \vec{\lambda}\cdot \vec{v}_+)+\Theta(p_-\ \vec{\lambda}\cdot \vec{v}_-)&\geq \Theta((2p-1)\ \vec{\lambda}\cdot\vec{z})\\
    p_+\ \Theta(\vec{\lambda}\cdot \vec{v}_+)+p_-\ \Theta(\vec{\lambda}\cdot \vec{v}_-)&\geq (2p-1)\ \Theta(\vec{\lambda}\cdot\vec{z})\\
    \tilde{\rho}_{\vec{x}}(\vec{\lambda})&\geq 0 \, .
\end{align}
Note that we have used $p_+,p_-\geq 0$ and $(2p-1)\geq 0$.\\

\textit{(ii) symmetric:}\\
From $\Theta(a)-\Theta(-a)=a$ we conclude:
\begin{align}
    a+b&=(a+b)\\
    \Theta(a)-\Theta(-a)+\Theta(b)-\Theta(-b)&=\Theta(a+b)-\Theta(-a-b)\\
    \Theta(a)+\Theta(b)-\Theta(a+b)&=\Theta(-a)+\Theta(-b)-\Theta(-a-b) \, .
\end{align}\\
Now we can choose $a= p_+\ \vec{\lambda}\cdot \vec{v}_+$ and $b= p_-\ \vec{\lambda}\cdot \vec{v}_-$ such that $a+b=p_+\ \vec{\lambda}\cdot \vec{v}_++p_-\ \vec{\lambda}\cdot \vec{v}_-=(2p-1)\ \vec{\lambda}\cdot\vec{z}$. We obtain directly:
\begin{align}
    \tilde{\rho}_{\vec{x}}(\vec{\lambda})=\frac{1}{\pi}(\Theta(a)+\Theta(b)-\Theta(a+b))=\frac{1}{\pi}(\Theta(-a)+\Theta(-b)-\Theta(-a-b))=\tilde{\rho}_{\vec{x}}(-\vec{\lambda}) \, .
\end{align}\\

\textit{(iii) area:}\\
Since $\int_{S_2} \Theta(\vec{\lambda}\cdot \vec{v})/\pi \ \mathrm{d}\vec{\lambda}=1$ (see Eq.~\eqref{normalization}) we obtain by linearity:
\begin{align}
    \int_{S_2} \tilde{\rho}_{\vec{x}}(\vec{\lambda}) \ \mathrm{d}\vec{\lambda}=&\int_{S_2} \frac{p_+}{\pi}\ \Theta(\vec{\lambda}\cdot \vec{v}_+) \ \mathrm{d}\vec{\lambda}+\int_{S_2} \frac{p_-}{\pi}\ \Theta(\vec{\lambda}\cdot \vec{v}_-) \ \mathrm{d}\vec{\lambda}-\int_{S_2} \frac{(2p-1)}{\pi}\ \Theta(\vec{\lambda}\cdot \vec{z}) \ \mathrm{d}\vec{\lambda}\\
    =&p_++p_--(2p-1)=1-(2p-1)=2(1-p) \, .
\end{align}\\

\textit{(iv) 1st bound:}\\
We can use $\Theta(a)+\Theta(b)\geq \Theta(a+b)$ with $a= p_+\ \vec{\lambda}\cdot \vec{v}_++p_-\ \vec{\lambda}\cdot \vec{v}_-=(2p-1)\ \vec{\lambda}\cdot\vec{z}$ and $b=- p_-\ \vec{\lambda}\cdot \vec{v}_-$. We obtain $a+b=p_+\ \vec{\lambda}\cdot \vec{v}_+$ and therefore:
\begin{align}
    \Theta((2p-1)\ \vec{\lambda}\cdot\vec{z})+\Theta(- p_-\ \vec{\lambda}\cdot \vec{v}_-)&\geq \Theta(p_+\ \vec{\lambda}\cdot \vec{v}_+)\\
    \Theta((2p-1)\ \vec{\lambda}\cdot\vec{z})+\Theta(- p_-\ \vec{\lambda}\cdot \vec{v}_-)+\Theta(p_-\ \vec{\lambda}\cdot \vec{v}_-)&\geq \Theta(p_+\ \vec{\lambda}\cdot \vec{v}_+)+\Theta(p_-\ \vec{\lambda}\cdot \vec{v}_-)\\
    p_-\ \Theta(- \vec{\lambda}\cdot \vec{v}_-)+p_-\ \Theta(\vec{\lambda}\cdot \vec{v}_-)&\geq p_+\ \Theta(\vec{\lambda}\cdot \vec{v}_+)+p_-\ \Theta(\vec{\lambda}\cdot \vec{v}_-)-(2p-1)\ \Theta(\vec{\lambda}\cdot\vec{z})\\
    p_-\ |\vec{\lambda}\cdot \vec{v}_-|&\geq \pi \cdot \tilde{\rho}_{\vec{x}}(\vec{\lambda}) \, .
\end{align}
If we choose $b=- p_+\ \vec{\lambda}\cdot \vec{v}_+$ instead, we obtain $p_+\ |\vec{\lambda}\cdot \vec{v}_+|\geq \pi \cdot \tilde{\rho}_{\vec{x}}(\vec{\lambda})$.\\

\textit{(v) 2nd bound:}\\
Here we prove:
\begin{align}
    \tilde{\rho}_{\vec{x}}(\vec{\lambda})\leq \frac{1}{2\pi}\frac{1-C^2}{\sqrt{1-C^2\sin^2{(\theta)}}+C|\cos{(\theta)}|}
\end{align}
where $C:=2p-1$ and $\cos{(\theta)}=\vec{\lambda}\cdot \vec{z}$.

We prove it in the following way: For a given vector, $\vec{\lambda} \in S_2$ we want to find the distribution $\tilde{\rho}_{\vec{x}}$ for which $\tilde{\rho}_{\vec{x}}(\vec{\lambda})$ is maximal. First, we focus on a vector $\vec{\lambda}$ in the lower hemisphere ($\vec{\lambda}\cdot \vec{z}\leq 0$). In that region, it turns out that $\tilde{\rho}_{\vec{x}}(\vec{\lambda})=\frac{1}{\pi}(p_+ \Theta(\vec{v}_+\cdot \vec{\lambda})+p_- \Theta(\vec{v}_-\cdot \vec{\lambda}))=\rho_{\vec{x}}(\vec{\lambda})$ and furthermore, only one of the two terms ($p_+ \vec{v}_+\cdot \vec{\lambda}$ or $p_- \vec{v}_-\cdot \vec{\lambda}$) is positive (if both are positive, we have $p_+ \vec{v}_+\cdot \vec{\lambda}+p_- \vec{v}_-\cdot \vec{\lambda}=C\ \vec{z}\cdot \vec{\lambda}>0$ which is a contradiction). Therefore, we want to find, for a given vector $\vec{\lambda}$, $\vec{v}_+$ and $p_+$ such that  $\rho_{\vec{x}}(\vec{\lambda})=\frac{1}{\pi}(p_+ \vec{v}_+\cdot \vec{\lambda})$ is maximal (we choose "+" w.l.o.g.).

For what follows, we choose the following parametrization: $\vec{v}_\pm=(\sin{(\alpha_\pm)}, 0, \cos{(\alpha_\pm)})^T$, $\vec{\lambda}=(\sin{(\theta)}, 0, \cos{(\theta)})^T$. Note that the maximum is the same for two different $\vec{\lambda}$ with the same $z$-component due to the rotational symmetry around the $z$-axis. This allows us to focus only on the particular choice of $\vec{\lambda}=(\sin{(\theta)}, 0, \cos{(\theta)})^T$ where we set the $y$-component to zero. Then the vector $\vec{v}_+$ that achieves the maximum will have zero $y$-component as well since we want to maximize the inner product between these two vectors. Solving the equation $p_+ \vec{v}_++p_- \vec{v}_-=C\ \vec{z}$ together with $p_++p_-=1$ leads to:
\begin{align}
    p_+=\frac{1-C^2}{2-2C\cos{(\alpha_+)}}\, , &&\sin{(\alpha_-)}=\frac{(1-C^2)\sin{(\alpha_+)}}{2C\cos{(\alpha_+)}-(1+C^2)}\, , &&\cos{(\alpha_-)}=\frac{(1+C^2)\cos{(\alpha_+)}-2C}{2C\cos{(\alpha_+)}-(1+C^2)} \, .
\end{align}
Here, $\sin{(\alpha_-)}$ and $\cos{(\alpha_-)}$ are stated merely for completeness and are not necessary for what follows. In order to maximize $\frac{1}{\pi}(p_+ \vec{v}_+\cdot \vec{\lambda})$, we have to find the maximal $\alpha_+$ for the function:
\begin{align}
    \frac{1}{\pi}(p_+ \vec{v}_+\cdot \vec{\lambda})=\frac{(1-C^2)\cos{(\theta-\alpha_+)}}{2\pi(1-C\cos{(\alpha_+)})} \, .
\end{align}
Maximizing over $\alpha_+$ leads to the condition $C\sin{(\theta)}=\sin{(\theta -\alpha_+)}$ or $\alpha_+=\theta-\arcsin{(C\sin{(\theta)})}$. For the two expressions in the above function that contain $\alpha_+$, we obtain $\cos{(\alpha_+)}=\cos{(\theta)}\sqrt{1-C^2\sin^2{(\theta)}}+C\sin^2{(\theta)}$ and $\cos{(\theta -\alpha_+)}=\sqrt{1-C^2\sin^2{(\theta)}}$. This leads to the following bound:
\begin{align}
    \frac{1}{\pi}(p_+ \vec{v}_+\cdot \vec{\lambda}) \leq \frac{1}{2\pi}\frac{1-C^2}{\sqrt{1-C^2\sin^2{(\theta)}}-C\cos{(\theta)}}=\frac{1}{2\pi}\frac{1-C^2}{\sqrt{1-C^2\sin^2{(\theta)}}+C|\cos{(\theta)}|} \, .
\end{align}
In the second step, we used that $\cos{(\theta)}\leq 0$ and therefore $\cos{(\theta)}=-|\cos{(\theta)}|$. For a vector in the upper hemisphere, we simply observe that the function $\tilde{\rho}_{\vec{x}}(\vec{\lambda})$ is symmetric $\tilde{\rho}_{\vec{x}}(\vec{\lambda})=\tilde{\rho}_{\vec{x}}(-\vec{\lambda})$. Since $-\vec{\lambda}=(-\sin{(\theta)}, 0, -\cos{(\theta)})^T$ this leads directly to
\begin{align}
    \tilde{\rho}_{\vec{x}}(\vec{\lambda})\leq \frac{1}{2\pi}\frac{1-C^2}{\sqrt{1-C^2\sin^2{(\theta)}}+C|\cos{(\theta)}|}\, ,
\end{align}
since the bound is invariant under changing the sign of $\sin{(\theta)}$ and $\cos{(\theta)}$.\\

\textit{(vi) 3rd bound:}\\
We maximize the 2nd bound over $\theta$. The maximum is reached when $\theta=\pi/2$, which leads to:
\begin{align}
    \tilde{\rho}_{\vec{x}}(\vec{\lambda})\leq \frac{\sqrt{1-C^2}}{2\pi}=\frac{\sqrt{p(1-p)}}{\pi} \, .
\end{align}
Note, that this bound is strictly weaker than the second bound but easier to state and useful for pedagogical reasons.
\end{proof}

\section{Improved one bit protocol}\label{improvedprot}
We can improve the protocol from the main text in two independent ways. The first improvement comes from the fact that $p_A(c=1|\vec{x},\vec{\lambda}_1)=(4\pi)\cdot \tilde{\rho}_{\vec{x}} (\vec{\lambda}_1)\leq 4 \sqrt{p(1-p)}$, where we used the bound $\tilde{\rho}_{\vec{x}}(\vec{\lambda})\leq \sqrt{p(1-p)}/\pi$. Hence, if the state is very weakly entangled ($p\lesssim 1$), the probability that Alice sends the bit $c=1$ is always small. Intuitively speaking, this allows us to rewrite the protocol into a form where Alice and Bob only communicate in a fraction of rounds, but in those rounds with a higher (rescaled) probability $p_A(c=1|\vec{x},\vec{\lambda}_1)\propto \tilde{\rho}_{\vec{x}} (\vec{\lambda}_1)$. In the limit where $p$ approaches one (the separable state $\ketbra{00}$), the fraction of rounds in which they have to communicate at all approaches even zero. The second improvement comes from using a better bound for the function $\tilde{\rho}_{\vec{x}}(\vec{\lambda})$. Indeed, the bound $\tilde{\rho}_{\vec{x}}(\vec{\lambda})\leq \sqrt{p(1-p)}/\pi$ is easy to state but not optimal. More precisely, we have proven in Appendix~\ref{rhotilde}:
\begin{align}
    0\leq \tilde{\rho}_{\vec{x}}(\vec{\lambda}) \leq \tilde{\rho}_{max}(\vec{\lambda}):=\frac{1}{2\pi}\frac{1-C^2}{\sqrt{1-C^2\sin^2{(\theta)}}+C|\cos{(\theta)}|} \, .
\end{align}
Here, $C:=2p-1$ and $\cos{(\theta)}=\vec{\lambda}\cdot \vec{z}$ is the $z$ component of $\vec{\lambda}$ in spherical coordinates. Note that, neither $\tilde{\rho}_{\vec{x}}(\vec{\lambda})$ nor $\tilde{\rho}_{max}(\vec{\lambda})$ are normalized, and we define the function $N(p)$ as the normalization of that function $\tilde{\rho}_{max}(\vec{\lambda})$:
\begin{align}
    N(p):=&\int_{S_2} \tilde{\rho}_{max}(\vec{\lambda}) \ \mathrm{d}\vec{\lambda} \, .
\end{align}
To not scare off the reader, we evaluate the integral after we present the protocol.

\begin{protocol}[$0.835\leq p\leq 1$, 1 bit in the worst case, $N(p)$ bits on average] \label{improvedonebitprotocol}
Same as Protocol~\ref{protocolgenfram} with the following 2. Step:\\
Alice and Bob share two random vectors $\vec{\lambda}_1, \vec{\lambda}_2 \in S_2$ according to the distribution:
    \begin{align}
    \rho(\vec{\lambda}_1)=\frac{1}{N(p)} \cdot \tilde{\rho}_{max}(\vec{\lambda}_1)\, ,&&
    \rho(\vec{\lambda}_2)=\frac{1}{\pi}\ \Theta(\vec{\lambda}_2 \cdot \vec{z}) \, .
\end{align}
In addition, Alice and Bob share a random bit $r$ distributed according to $p(r=0)=1-N(p)$ and $p(r=1)=N(p)$. If $r=0$, Alice and Bob do not communicate and both set $\vec{\lambda}:=\vec{\lambda}_2$. If $r=1$, Alice sets $c=1$ with probability:
    \begin{align}
    \begin{split}
        p(c=1|\vec{x},\vec{\lambda}_1)=\tilde{\rho}_{\vec{x}} (\vec{\lambda}_1)/\tilde{\rho}_{max}(\vec{\lambda}_1)
    \end{split}
    \end{align}
    and otherwise she sets $c=2$. She communicates the bit $c$ to Bob. Both set $\vec{\lambda}:=\vec{\lambda}_c$ and reject the other vector.
\end{protocol}
\begin{proof}
    We show that the distribution of the shared vector $\vec{\lambda}$ becomes exactly the required $\rho_{\vec{x}}(\vec{\lambda})$. To see that, consider all the cases where Alice chooses the first vector. This happens only when $r=1$ and when she sets the bit $c$ to 1. This samples the distribution:
\begin{align}
    p(r=1)\cdot p(c=1|\vec{x},\vec{\lambda}_1)\cdot \rho(\vec{\lambda}_1)=\tilde{\rho}_{\vec{x}}(\vec{\lambda}_1) \, .
\end{align}
The total probability that she is choosing the first vector is $\int_{S_2} p(r=1)\cdot p(c=1|\vec{x},\vec{\lambda}_1)\cdot \rho(\vec{\lambda}_1) \, \mathrm{d}\vec{\lambda}_1 =\int_{S_2} \tilde{\rho}_{\vec{x}}(\vec{\lambda}_1) \, \mathrm{d}\vec{\lambda}_1 = 2(1-p)$. In all the remaining cases, with total probability $2p-1$, she is choosing vector $\vec{\lambda}:=\vec{\lambda}_2$. Therefore, the total distribution of the chosen vector $\vec{\lambda}$ becomes the desired distribution
\begin{align}
    \tilde{\rho}_{\vec{x}}(\vec{\lambda})+\frac{(2p-1)}{\pi}\ \Theta(\vec{\lambda}\cdot \vec{z})=\rho_{\vec{x}}(\vec{\lambda}) \, .
\end{align}
Here, the first term $\tilde{\rho}_{\vec{x}}(\vec{\lambda})$ in the sum corresponds to all the instances where Alice chooses $\vec{\lambda}:=\vec{\lambda}_1$ and the second term to all the instances, where she chooses $\vec{\lambda}_2$. In order for the protocol to be well defined, it has to hold that $0\leq p(c=1|\vec{x},\vec{\lambda}_1)\leq 1$ and $0\leq N(p)\leq 1$. The first bound is true since $0\leq \tilde{\rho}_{\vec{x}} (\vec{\lambda}_1)\leq \tilde{\rho}_{max} (\vec{\lambda}_1)$ and the second bound $0\leq N(p)\leq 1$ holds whenever $0.835\leq p\leq 1$.
\end{proof}

We can explicitly solve that integral for $N(p)$ by using spherical coordinates $\vec{\lambda}=(\sin{\theta}\cdot \cos{\phi}, \sin{\theta}\cdot \sin{\phi}, \cos{\theta})$:
\begin{align}
\begin{split}
    N(p)=&\int_{S_2} \tilde{\rho}_{max}(\vec{\lambda}) \ \mathrm{d}\vec{\lambda} 
    = \int^{2\pi}_{0} \int^{\pi}_{0} \tilde{\rho}_{max}(\theta,\phi) \cdot \sin{\theta}\  \mathrm{d}\theta  \ \mathrm{d}\phi \\
    =& \int^{2\pi}_{0} \int^{\pi}_{0} \frac{1}{2\pi}\frac{1-C^2}{\sqrt{1-C^2\sin^2{(\theta)}}+C|\cos{(\theta)}|} \cdot \sin{\theta}\  \mathrm{d}\theta  \ \mathrm{d}\phi \\
    =& \int^{\pi}_{0} \frac{1-C^2}{\sqrt{1-C^2\sin^2{(\theta)}}+C|\cos{(\theta)}|} \cdot \sin{\theta}\  \mathrm{d}\theta  \\
    =& 2 \int^{\pi/2}_{0} \frac{1-C^2}{\sqrt{1-C^2\sin^2{(\theta)}}+C \cos{(\theta)}} \cdot \sin{\theta}\  \mathrm{d}\theta \, .
\end{split}
\end{align}
Substituting $x=C \cdot \cos{\theta}$ leads to:
\begin{align}
\begin{split}
    N(p)=& \frac{2(1-C^2)}{C} \int^{C}_{0} \frac{1}{\sqrt{(1-C^2)+x^2}+x}\  \mathrm{d}x  \\
    =& \frac{2(1-C^2)}{C} \left[\frac{1}{2}\log{\left(\sqrt{(1-C^2)+x^2}+x\right)}-\frac{1-C^2}{4\left(\sqrt{(1-C^2)+x^2}+x\right)^2} \right]^{C}_{0} \\
    =& \frac{1-C^2}{2C} \left(\log{\left(\frac{1+C}{1-C}\right)}+\frac{2C}{1+C} \right)\\
    =& \frac{1-C^2}{2C} \log{\left(\frac{1+C}{1-C}\right)}+(1-C) \\
    =& \frac{2p(1-p)}{2p-1} \log{\left(\frac{p}{1-p}\right)}+2(1-p) \, .
\end{split}
\end{align}
Here we used $C=2p-1$ in the last step.

\section{Maximal local content}\label{applocalcontent}
In Appendix~\ref{improvedprot}, we showed that it is possible to simulate weakly entangled states without communication in some fraction of rounds. One can even go one step further and maximize the fraction of rounds in which no communication is required. This is called the local content of the state $\ket{\Psi_{AB}}$. More formally, a simulation of an entangled qubit pair can be decomposed into a local part $p_L(a,b|\vec{x},\vec{y})$ that can be implemented by Alice and Bob without communication and the remaining non-local content, denoted as $p_{NL}(a,b|\vec{x},\vec{y})$:
\begin{align}
    p_Q(a,b|\vec{x},\vec{y})=p_L \cdot p_L(a,b|\vec{x},\vec{y})+(1-p_L)\cdot  p_{NL}(a,b|\vec{x},\vec{y}) \, . \label{localcontent}
\end{align}
The problem consists in finding the maximal value of $p_L$ for a given state $\ket{\Psi_{AB}}=\sqrt{p}\ket{00}+\sqrt{1-p}\ket{00}$, that we denote here as $p^{max}_L(p)$. 
For a maximally entangled state, Elitzur, Popescu, and Rohrlich showed that the local content is necessarily zero (hence $p^{max}_L(p=1/2)=0$), also known as the EPR2 decomposition \cite{elitzur1992} (see also Barrett et al.~\cite{barrett2006}). For general entangled qubit pairs, it was shown by Scarani~\cite{Scarani2008localpart} that the local content is upper bounded by $2p-1$, hence $p^{max}_L(p)\leq 2p-1$. At the same time, subsequently better lower bounds were found \cite{elitzur1992, Scarani2008localpart, Branciard2010localpart, Portmann2012localpart}. Finally, Portmann et al. \cite{Portmann2012localpart} found an explicit decomposition with a local content of $p_L(p)=2p-1$, hence proving that the upper and lower bound coincide and, therefore, $p^{max}_L(p)=2p-1$. Here, we give an independent proof of the result by Portmann et al. \cite{Portmann2012localpart}. More precisely, we provide a protocol that simulates any pure entangled two-qubit state of the general form $\ket{\Psi_{AB}}=\sqrt{p}\ket{00}+\sqrt{1-p}\ket{00}$ with a local content of $p_L=2p-1$.\\

We remark that often in the literature, the state $\ket{\Psi_{AB}}$ is written as $\ket{\Psi_{AB}}=\cos{\theta}\ket{00}+\sin{\theta}\ket{00}$ where $\cos{\theta}\geq \sin{\theta}$. These two notations are related through the following expressions: $\cos{2\theta}=2p-1$ which follows from $\cos{\theta}=\sqrt{p}$ and $\sin{\theta}=\sqrt{1-p}$ together with $\cos{2\theta}=\cos^2{\theta}-\sin^2{\theta}=p-(1-p)=2p-1$.

\begin{protocol}[$1/2\leq p\leq 1$, maximal local content of $p_L=2p-1$]
Same as Protocol~\ref{protocolgenfram} with the following 2. Step:\\
Alice and Bob share a random vector $\vec{\lambda}_1 \in S_2$ according to the distribution ($\vec{z}:=(0,0,1)^T$):
    \begin{align}
    \rho(\vec{\lambda}_1)=\frac{1}{\pi}\ \Theta(\vec{\lambda}_1 \cdot \vec{z}) \, .
\end{align}
In addition, Alice and Bob share a random bit $r$ distributed according to $p(r=0)=2p-1$ and $p(r=1)=2(1-p)$. If $r=0$, Alice and Bob do not communicate and both set $\vec{\lambda}:=\vec{\lambda}_1$. If $r=1$, Alice samples $\vec{\lambda}\in S_2$ according to the distribution $\tilde{\rho}_{\vec{x}}(\vec{\lambda})$ and communicates that $\vec{\lambda}$ to Bob. (Alice can for example encode the three coordinates of $\vec{\lambda}$ and sends this information to Bob.)
\end{protocol}

\begin{proof}
    Whenever $r=1$, the resulting distribution of $\vec{\lambda}$ is, by construction, $\tilde{\rho}_{\vec{x}}(\vec{\lambda})$. Whenever, $r=0$ (with probability $2p-1$) the distribution of $\vec{\lambda}$ is $\Theta(\vec{\lambda} \cdot \vec{z})/\pi$. Therefore, the total distribution for the chosen vector $\vec{\lambda}$ becomes the desired distribution
\begin{align}
    \tilde{\rho}_{\vec{x}}(\vec{\lambda})+\frac{(2p-1)}{\pi}\ \Theta(\vec{\lambda}\cdot \vec{z})=\rho_{\vec{x}}(\vec{\lambda}) \, .
\end{align}
Hence, by the proof of Protocol~\ref{protocolgenfram}, this exactly reproduces quantum correlations.
\end{proof}

This directly provides a decomposition of the above form (Eq.~\eqref{localcontent}) for $p_L=2p-1$. More precisely, whenever $r=0$ (with total probability $2p-1$), they do not communicate and, hence, implement a local strategy $p_L(a,b|\vec{x},\vec{y})$. On the other hand, whenever $r=0$ they communicate and, therefore, implement a non-local behavior $p_{NL}(a,b|\vec{x},\vec{y})$. That protocol requires an unbounded amount of communication in the rounds where $r=1$ and there are more efficient ways to sample the distribution $\tilde{\rho}_{\vec{x}}(\vec{\lambda})$. However, this is not an issue for determining $p_L^{max}(p)$ since we only maximize the local content, hence, the number of rounds in which Alice and Bob do not have to communicate. At the same time, we are not concerned with the amount of communication in the remaining rounds. One can also ask, what is the maximal local content under the restriction that Alice and Bob only communicate a single bit in the remaining rounds. We found in Appendix~\ref{improvedprot} such a decomposition. However, in general, it seems unlikely that a decomposition that attains the maximal local content of $p^{max}_L(p)=2p-1$ can be achieved if the two parties communicate only a single bit in the remaining rounds.
\end{document}